\documentclass[11pt,conference,a4paper, onecolumn]{IEEEtran}
\usepackage[utf8]{inputenc}
\usepackage{amssymb, amsthm}
\usepackage[cmex10]{amsmath}
\usepackage[usenames,dvipsnames]{xcolor}
\usepackage{hyperref}
\hypersetup{colorlinks=false, pdfborderstyle={/S/U/W 1},pdfborder=0 0 1, citebordercolor=Blue, filebordercolor=Red, linkbordercolor=Red, urlbordercolor=Blue}
\usepackage[ruled,vlined,titlenumbered]{algorithm2e}
\usepackage[english]{babel}
\usepackage{cite}
\newtheorem{definition}{Definition}
\newtheorem{theorem}[definition]{Theorem}
\newtheorem{lemma}[definition]{Lemma}

\newtheorem{remark}[definition]{Remark}
\newtheorem{example}[definition]{Example}

\DeclareMathOperator{\lcm}{lcm}

\newcommand{\gcdab}[2]{\ensuremath{\gcd(#1,#2)}}
\newcommand{\lcmab}[2]{\ensuremath{\lcm(#1,#2)}}
\DeclareMathOperator{\defi}{def}
\newcommand{\defeq}{\overset{\defi}{=}}
\newcommand{\defequiv}{\overset{\defi}{\equiv}}
\newcommand{\F}[1]{\mathbb F_{#1}}
\newcommand{\Fq}{\F{q}}
\newcommand{\Fxsub}[1]{\ensuremath{\mathbb{F}_{#1}[X]}}
\newcommand{\Fqx}{\Fxsub{q}}

\newcommand{\Z}{\ensuremath{\mathbb{Z}}}
\newcommand{\SET}[1]{\ensuremath{\mathsf{#1}}}
\newcommand{\LIN}[4]{\ensuremath{[#1,#2,#3]_{#4}}}
\newcommand{\LINq}[3]{\ensuremath{[#1,#2,#3]_{q}}}
\newcommand{\defset}[2][\empty]{
  \ifthenelse{\equal{#1}{\empty}}
    {\ensuremath{\SET{D}_{#2}}}
    {\ensuremath{\SET{D}^{[#1]}_{#2}}}
}
\newcommand{\CYCa}{\ensuremath{\mathcal{A}}}
\newcommand{\CYCan}{\ensuremath{n_a}} 
\newcommand{\CYCak}{\ensuremath{k_a}} 
\newcommand{\CYCad}{\ensuremath{d_a}} 
\newcommand{\CYCas}{\ensuremath{l_a}} 
\newcommand{\CYCb}{\ensuremath{\mathcal{B}}}
\newcommand{\CYCbn}{\ensuremath{n_b}} 
\newcommand{\CYCbk}{\ensuremath{k_b}} 
\newcommand{\CYCbd}{\ensuremath{d_b}} 
\newcommand{\CYCbs}{\ensuremath{l_b}} 
\newcommand{\CYC}{\ensuremath{\mathcal{C}}}
\newcommand{\CYCn}{\ensuremath{n}} 
\newcommand{\CYCk}{\ensuremath{k}} 
\newcommand{\CYCd}{\ensuremath{d}} 
\newcommand{\CYCprime}{\ensuremath{\overline{\mathcal{C}}}}
\newcommand{\CYCnprime}{\ensuremath{\overline{n}}} 
\newcommand{\CYCkprime}{\ensuremath{\overline{k}}} 
\newcommand{\CYCdprime}{\ensuremath{\overline{d}}} 
\newcommand{\mult}{\ensuremath{s}} 
\newcommand{\CYCs}{\ensuremath{l}} 
\newcommand{\inta}{\ensuremath{a}}
\newcommand{\intb}{\ensuremath{b}}
\newcommand{\HTconst}{\ensuremath{f}}
\newcommand{\HTconsta}{\ensuremath{f_a}}
\newcommand{\HTconstb}{\ensuremath{f_b}}
\newcommand{\lenseq}{\ensuremath{\delta}}
\newcommand{\noseq}{\ensuremath{\nu}}
\newcommand{\HTmult}{\ensuremath{m}}
\newcommand{\HTmulta}{\ensuremath{m_a}}
\newcommand{\HTmultb}{\ensuremath{m_b}}
\newcommand{\HasseDer}[2]{\ensuremath{#1^{[#2]}}} 
\newcommand{\Der}[2]{\ensuremath{#1^{(#2)}}} 
\newcommand{\ErrorSet}{\SET{E}} 
\newcommand{\NoErrors}{\ensuremath{\varepsilon}} 
\newcommand{\SupportSecond}{\SET{Y}}
\newcommand{\BoundOne}{\ensuremath{d_{\textsf{I}}}}
\newcommand{\BoundTwo}{\ensuremath{d_{\textsf{II}}}}
\newcommand{\BoundThree}{\ensuremath{d_{\textsf{III}}}}
\renewcommand{\tilde}{\widetilde} 
\renewcommand{\bar}{\overline} 
\newcommand{\brac}[2][\empty]{ 
  \ifthenelse{\equal{#1}{\empty}}
    {\ensuremath{\left(#2\right)}}
    {\ensuremath{#1\left(#2\right)}}
}
\newcommand{\binomsmall}[2]{\genfrac{(}{)}{0pt}{1}{#1}{#2}} 
\newcommand{\cupmax}{\stackrel{\text{max}}{\cup}} 
\newcommand{\inter}[1]{\ensuremath{[#1)}} 
\begin{document}
\title{Decoding of Repeated-Root Cyclic Codes up to New Bounds on Their Minimum Distance}
\IEEEoverridecommandlockouts
\author{\IEEEauthorblockN{Alexander Zeh}\thanks{This work has been supported by the German research council (Deutsche Forschungsgemeinschaft, DFG) under grants Bo867/22-1 and Ze1016/1-1 and was initiated when both authors were affiliated with the Institute of Communications Engineering, University of Ulm, Ulm, Germany.}
\IEEEauthorblockA{Computer Science Department\\
Technion, Haifa, Israel\\
\texttt{alex@codingtheory.eu}
}
\and
\IEEEauthorblockN{Markus Ulmschneider}
\IEEEauthorblockA{Institute of Communications and Navigation\\
German Aerospace Center (DLR), Germany\\
\texttt{markus.ulmschneider@dlr.de}}
}
\maketitle

\begin{abstract}
The well-known approach of Bose, Ray-Chaudhuri and Hocquenghem and its generalization by Hartmann and Tzeng are lower bounds on the minimum distance of simple-root cyclic codes. We generalize these two bounds to the case of repeated-root cyclic codes and present a syndrome-based burst error decoding algorithm with guaranteed decoding radius based on an associated folded cyclic code.

 Furthermore, we present a third technique for bounding the minimum Hamming distance based on the embedding of a given repeated-root cyclic code into a repeated-root cyclic product code. A second quadratic-time probabilistic burst error decoding procedure based on the third bound is outlined.
\end{abstract}


\section{Introduction}
The length of a conventional linear cyclic block code $\CYC$ over a finite field $\Fq$ has to be co-prime to the field characteristic $p$. This guarantees that the generator polynomial of $\CYC$ has roots of multiplicity at most one and therefore we refer to these codes as simple-root cyclic codes. The approach of Bose and Ray-Chaudhuri and Hocquenghem (BCH,~\cite{bose_class_1960, hocquenghem_codes_1959}) and of Hartmann and Tzeng (HT,~\cite{hartmann_decoding_1972, hartmann_generalizations_1972}) gives a lower bound on the minimum distance of simple-root cyclic codes. Both approaches are based on consecutive sequences of roots of the generator polynomial.
We give---similar to the BCH and the HT bound---two lower bounds on the minimum Hamming distance of a repeated-root cyclic code, i.e., a cyclic code whose length is not relatively co-prime to the characteristic $p$ of the field $\Fq$ and therefore its generator polynomial can have roots with multiplicities greater than one.

Repeated-root cyclic codes were first investigated by Berman~\cite{berman_semisimple_1967}. A special class of Maximum Distance Separable (MDS) repeated-root constacyclic codes was treated by Massey~\textit{et al.} in~\cite{massey_polynomial_1973, massey_hasse_1986} and the advantages of a syndrome-based decoding were outlined. An alternative derivation of the minimum Hamming distance of these repeated-single-root MDS codes and their application to secret-key cryptosystems was given by da Rocha in~\cite{da_rocha_jr._repeated-single-root_1994}.
Castagnoli~\textit{et al.}~\cite{castagnoli_minimum_1989, castagnoli_asymptotic_1989, castagnoli_repeated-root_1991} gave an elaborated description of repeated-root cyclic codes including the explicit construction of the parity-check matrix, which was investigated for the case $q=2$ slightly earlier by Latypov~\cite{latypov_checking_1988}. Although the asymptotic badness of repeated-root cyclic codes was shown in~\cite{castagnoli_minimum_1989, castagnoli_asymptotic_1989, castagnoli_repeated-root_1991}, several good binary repeated-root cyclic codes were constructed by van Lint in~\cite{van_lint_repeated-root_1991} with distances close to the Griesmer bound. Zimmermann~\cite{zimmermann_generalizations_1996} reproved some of Castagnoli's result by cyclic group algebra and Nedeloaia gave a squaring construction of all binary repeated-root cyclic codes in~\cite{nedeloaia_weight_2003}. Recent publications of Ling--Niederreiter--Sol\'{e}~\cite{ling_algebraic_2006} and Dinh~\cite{dinh_repeated-root_2012, dinh_structure_2013} consider repeated-root quasi-cyclic codes.

Besides the generalization of the BCH and the HT bound to repeated-root cyclic codes, we provide a third lower bound on the minimum Hamming distance. Similar to the approach~\cite{zeh_new_2012, zeh_decoding_2012} for simple-root cyclic codes, this bound is based on the embedding of a given repeated-root cyclic code into a repeated-root cyclic product code. Therefore, we recall the relevant theorems of Burton and Weldon~\cite{burton_cyclic_1965} and Lin and Weldon~\cite{lin_further_1970} for repeated-root cyclic product codes that are the basis for the proof of our third bound, which generalizes the results of our previous work on simple-root cyclic codes~\cite{zeh_new_2012, zeh_decoding_2012}. Moreover, we present two burst error decoding schemes based on the derived bounds.

The paper is structured as follows. In Section~\ref{sec_RepeatedRootCyclic}, we give necessary preliminaries for repeated-root cyclic codes and introduce our notation. Section~\ref{sec_SimpleBounds} provides the generalizations of the BCH and the HT bound, which are denoted by $\BoundOne$ and $\BoundTwo$ respectively, and in addition a syndrome-based error-correction algorithm with guaranteed decoding radius. The defining set of a repeated-root cyclic product code is given explicitly in Section~\ref{sec_DefiningSet}, which is necessary to prove our third bound $\BoundThree$ on the minimum Hamming distance of a repeated-root cyclic code in Section~\ref{sec_BoundsProduct}. Section~\ref{sec_Decoding} gives a probabilistic burst error decoding approach based on the Generalized Extended Euclidean Algorithm (GEEA,~\cite{feng_generalized_1989}). We conclude this paper in Section~\ref{sec_Conclusion}.

\section{Repeated-Root Cyclic Codes} \label{sec_RepeatedRootCyclic}
\subsection{Notation and Preliminaries}
Let $q$ be a power of a prime $p$. $\Fq$ denotes the finite field of order $q$ and characteristic $p$ and $\Fqx$ the polynomial ring over $\Fq$ with indeterminate $X$. Let $n$ be a positive integer and denote by $\inter{n}$ the set of integers $\{0,1,\dots, n-1\}$. A vector of length $n$ is denoted by a lowercase bold letter as $\mathbf{v} = (v_0 \, v_1 \, \dots \, v_{n-1})$. 
A set is denoted by a capital letter sans serif like $\SET{D}$.

A linear $\LINq{\CYCn}{\CYCk}{\CYCd}$ code over $\Fq$ of length $\CYCn$, dimension $\CYCk$ and  minimum Hamming distance $\CYCd$ is denoted by a calligraphic letter like $\CYC$.

Let us recapitulate the definition of the Hasse derivative~\cite{hasse_theorie_1936} in the following. Let $a(X) = \sum_i a_i X^i$ be a polynomial in $\Fqx$, then the $j$-th Hasse derivative is:
\begin{equation} \label{eq_defHasseDerivative}
\HasseDer{a}{j}(X) \defeq \sum_i \binom{i}{j} a_i X^{i-j}.
\end{equation}
Let $\Der{a}{j}(X)$ denote the formal $j$-th derivative of $a(X)$. The fact that
$\Der{a}{j}(X) = j! \, \HasseDer{a}{j}(X) $ explains why the Hasse derivative is considered in fields with a prime characteristic $p$, because then $j! = 0$ and hence also $\Der{a}{j}(X) = 0$ for all $j \geq p$. 
We say a univariate polynomial $a(X) \in \Fqx$ with $\deg a(X) \geq \mult$ has a root at $\gamma$ with multiplicity $s$ if:
\begin{equation*}
\HasseDer{a}{j}(\gamma) = 0, \quad \forall j \in \inter{\mult}.
\end{equation*}

\subsection{Defining Set}
A linear $\LINq{\CYCnprime}{\CYCkprime}{\CYCdprime}$ simple-root cyclic code $\CYCprime$ over $\Fq$ with characteristic $p$ is an ideal in the ring $\Fqx / (X^{\CYCnprime}-1)$ generated by $\overline{g}(X)$, where $\gcd(\CYCnprime,p)=1$.
The generator polynomial $\overline{g}(X) \in \Fqx$ has roots with multiplicity at most one in the splitting field $\F{q^{\CYCs}}$, where $\CYCnprime \mid (q^{\CYCs} -1)$. 
A cyclotomic coset $ M_{i,\CYCnprime,q}$ is denoted by:
\begin{equation*} 
 M_{i,\CYCnprime,q} = \big\{ \, iq^j \mod \CYCnprime \; \vert \; j \in \inter{\CYCnprime_i} \big\},
\end{equation*}
where $\CYCnprime_i$ is the smallest integer such that $iq^{\CYCnprime_i} \equiv
i \mod \CYCnprime$. Let $\gamma$ be an element of order $\CYCnprime$ in $\F{q^{\CYCs}}$.
The minimal polynomial of the element $\gamma^i$ is: 
\begin{equation*}
M_{i,\CYCnprime,q}(X) = \prod_{j \in M_{i,\CYCnprime,q}} (X-\gamma^j).
\end{equation*}
Let $\gcd(\CYCnprime,p)=1$ and $n=p^{\mult} \CYCnprime$. A linear $\LINq{n}{k}{d}$ repeated-root cyclic code $\CYC$ is an ideal in the ring
\begin{equation*}
\Fqx / (X^n-1) = \Fqx / (X^{\CYCnprime}-1)^{p^{\mult}}.
\end{equation*}
The generator polynomial of an $\LINq{n}{k}{d}$ repeated-root cyclic code $\CYC$ is
\begin{equation*}
g(X) = \prod_i M_{i,\CYCnprime,q}(X)^{\mult_i},
\end{equation*}
where $\mult_i \leq p^{\mult} $.
The defining set $\defset{\CYC}$ of an $\LINq{\CYCn=p^{\mult} \CYCnprime}{\CYCk}{\CYCd}$ repeated-root cyclic code $\CYC$ with generator polynomial $g(X)$ is a set of tuples, where the first entry of the tuple is the index of a zero and the second its multiplicity, namely:
\begin{equation} \label{eq_definingset}
\defset{\CYC}  =  \Big\{ i^{\langle \mult_i \rangle}  \; | \; 0 \leq i \leq \CYCnprime-1,\quad \HasseDer{g}{j}(\gamma^i)=0, \quad \forall j \in \inter{\mult_i} \Big\},
\end{equation}
Furthermore, we introduce the following short-hand notation for a given $z \in \Z$:
\begin{equation} \label{eq_definingsetext}
\begin{split}
\defset[z]{\CYC}  & \defeq  \Big\{ (i+z)^{\langle \mult_i \rangle} \; | \; i^{\langle \mult_i \rangle} \in \defset{\CYC} \Big\}.
\end{split}
\end{equation}
For two given defining sets $\defset{\CYCa}$ and $\defset{\CYCb}$, define 
\begin{equation} \label{eq_cupmaxdefintion}
\defset{\CYCa} \cupmax \defset{\CYCb} \defeq \Big\{ i^{\langle \mult_i \rangle}   \; | \;  \mult_i = \max(a_i,b_i), \text{where} \; i^{\langle a_i \rangle}  \in  \defset{\CYCa} \; \text{and} \; i^{\langle b_i \rangle}  \in  \defset{\CYCb} \Big\}.
\end{equation}

\section{Two Bounds On the Minimum Hamming Distance of Repeated-Root Cyclic Codes And Burst Error Correction} \label{sec_SimpleBounds}
\subsection{Lower Bounds on the Minimum Hamming Distance} \label{subsec_BCHHTBounds}
In the following, we prove two lower bounds on the minimum Hamming distance of repeated-root cyclic codes. They generalize the well-known BCH~\cite{bose_class_1960, hocquenghem_codes_1959} and HT~\cite{hartmann_decoding_1972} approach suited for simple-root cyclic codes.
\begin{theorem}[Bound I: BCH-like Bound for a Repeated-Root Cyclic Code] \label{theo_BCHRR}
Let an $\LINq{\CYCn}{\CYCk}{\CYCd}$ repeated-root cyclic code $\CYC$ over $\Fq$ with characteristic $p$ and generator polynomial $g(X)$ with $\deg g(X) \geq p^{\mult}-1$ be given. Let $\CYCn=p^{\mult} \CYCnprime$, where \(\gcd(\CYCnprime, p) = 1\). Let \(\gamma\) be an element of order \(\CYCnprime\) in an extension field of $\Fq$. Furthermore, let three integers $\HTconst$, $\HTmult \neq 0$ and $\lenseq \geq 2$ with \(\gcd(\CYCnprime,\HTmult) = 1\) be given, such that for any codeword \(c(X) \in \CYC\)
\begin{equation} \label{eq_bch_proof_rr_basis}
 \sum_{i=0}^\infty \HasseDer{c}{p^{\mult}-1}(\gamma^{\HTconst+i\HTmult}) X^i \equiv 0 \mod X^{\lenseq-1}
\end{equation}
holds.
Then, the minimum distance of \(\CYC\) is at least $\BoundOne \defeq \lenseq$.
\end{theorem}
\begin{IEEEproof}
First, let us prove that the left-hand side of~\eqref{eq_bch_proof_rr_basis} cannot be zero. Assume it is the zero polynomial. Then, all \(\gamma^0,\gamma^1,\ldots,\gamma^{\CYCnprime-1}\) are roots of the codeword \(c(X)\) with multiplicity \(p^{\mult}\), yielding that \(\deg c(X)= p^{\mult} \CYCnprime = \CYCn\), which contradicts the fact that the degree of a codeword $c(X)$ of an $\LINq{n}{k}{d}$ code is smaller than $n$.
Second, we rewrite the expression left-hand side of~\eqref{eq_bch_proof_rr_basis} more explicitly. Let \(\SET{Y}=\{i: c_i \neq 0\}\) be the support of a non-zero codeword. 
We obtain:
\begin{align} \label{eq_proofBCHRR1}
\sum_{i=0}^\infty \HasseDer{c}{p^{\mult}-1}(\gamma^{\HTconst+i\HTmult}) X^i \nonumber
& =  \sum_{i=0}^\infty \sum_{u \in \SET{Y}} \binomsmall{u}{p^{\mult}-1}c_u {\brac{\gamma^{\HTconst+i\HTmult}}}^{u-p^{\mult}+1}  X^i \nonumber\\
& = \sum_{u \in \SET{Y}} \binomsmall{u}{p^{\mult}-1} c_u \gamma^{(u-p^{\mult}+1)\HTconst} \sum_{i=0}^\infty \brac{\gamma^{(u-p^{\mult}+1)\HTmult}X}^i.
\end{align}
With the geometric series, we get from~\eqref{eq_proofBCHRR1}:
\begin{align}
\sum_{u \in \SET{Y}} \binomsmall{u}{p^{\mult}-1} c_u \gamma^{(u-p^{\mult}+1)\HTconst} \sum_{i=0}^\infty \brac{\gamma^{(u-p^{\mult}+1)\HTmult}X}^i & = \sum_{u \in \SET{Y}} \binomsmall{u}{p^{\mult}-1}  c_u \gamma^{(u-p^{\mult}+1)\HTconst} \frac{1}{1-{\gamma^{(u-p^{\mult}+1)\HTmult}} X} \label{eq_BCHProof_Denominator}, 
\end{align}
and with 
\begin{equation} \label{eq_CommonDenominator}
 \sum_{i \in \SET{Y}} \frac{a_i}{1-Xb_i} = \frac{ \sum_{i \in \SET{Y}} a_i \frac{D}{1-Xb_i}}{D}, 
\end{equation}
where $D \defeq \lcm ( (1-Xb_i) : i \in \SET{Y}) $, we obtain from~\eqref{eq_BCHProof_Denominator}:
\begin{align}
\sum_{u \in \SET{Y}} \binomsmall{u}{p^{\mult}-1}  c_u \gamma^{(u-p^{\mult}+1)\HTconst} \frac{1}{1-{\gamma^{(u-p^{\mult}+1)\HTmult}} X} & = \frac{ \sum_{u \in \SET{Y}} \binomsmall{u}{p^{\mult}-1}  c_u \gamma^{(u-p^{\mult}+1)\HTconst} 
    \frac{\lcm \brac{  1-{\gamma^{(j-p^{\mult}+1)m}X}\ :\ {j \in \SET{Y} }  }}
     {  1-{\gamma^{(u-p^{\mult}+1)m}X}\   }} 
    {\lcm \brac{  1-{\gamma^{(i-p^{\mult}+1)m}X}\ :\ {i \in \SET{Y} }  }}\label{eq:bch_proof_rr1}\\
& \equiv 0 \mod X^{\lenseq-1}.\ \label{eq:bch_proof_rr2}
\end{align}
Obviously, the degree of the numerator of~\eqref{eq:bch_proof_rr1} cannot be greater than \(|\SET{Y}|-1\), and it cannot be smaller than \(\lenseq-1\), since \eqref{eq:bch_proof_rr2} must be fulfilled. Since this is true for all codewords, the minimum distance of \(\CYC\) is not smaller than \(|\SET{Y}|\). Thus, with \(|\SET{Y}| \geq \lenseq \) follows that the true minimum distance of \(\CYC\) is at least \(\lenseq \).
\end{IEEEproof}
Thm.~\ref{theo_BCHRR} tells us that a repeated-root cyclic code of length $\CYCn=p^{\mult} \CYCnprime$ with generator polynomial $g(X)$ that has $\delta-1$ consecutive zeros of highest multiplicity $p^{\mult}$, i.e., 
\begin{equation*}
\HasseDer{g}{p^{\mult}-1}(\gamma^{\HTconst}) = \HasseDer{g}{p^{\mult}-1}(\gamma^{\HTconst+\HTmult}) = \dots = \HasseDer{g}{p^{\mult}-1}(\gamma^{\HTconst +(\lenseq-2)\HTmult}) = 0,
\end{equation*}
has at least minimum distance $\delta$. If $s=0$, the repeated-root cyclic code is a simple-root cyclic code and then Thm.~\ref{theo_BCHRR} coincides with the BCH bound~\cite{bose_class_1960,hocquenghem_codes_1959}.
\begin{remark}[Parameters]
To obtain the parameters $\HTconst, \HTmult$ and $\lenseq$ as in Thm.~\ref{theo_BCHRR}, one needs to check the $(p^{\mult}-1)$th Hasse derivative of the given generator polynomial (respectively the defining set) of a given repeated-root cyclic code and find $\HTconst$ and $\HTmult$ that maximize $\lenseq$.
The advantage of the representation as in \eqref{eq_bch_proof_rr_basis} and in \eqref{eq_ht_proof_rr_basis} is that a syndrome definition can directly be obtained and an algebraic decoding algorithm can be formulated (see Section~\ref{subsec_decodingBoundOneTwo}).
\end{remark}

\begin{theorem}[Bound II: HT-like for a Repeated-Root Cyclic Code]  \label{theo_HTRR}
Let an $\LINq{\CYCn=p^{\mult} \CYCnprime}{\CYCk}{\CYCd}$ repeated-root cyclic code $\CYC$ over $\Fq$ with characteristic $p$ and generator polynomial $g(X)$ with $\deg g(X) \geq p^{\mult}-1$ be given, where \(\gcd(\CYCnprime, p) = 1\). Let \(\gamma\) be an element of order \(\CYCnprime\) in an extension field of $\Fq$. Furthermore, let four integers $\HTconst$, $\HTmult \neq 0$, $\lenseq \geq 2$ and $\noseq \geq 0$ with \(\gcd(\CYCnprime,\HTmult) = 1\) be given, such that for any codeword \(c(X) \in \CYC\)
\begin{equation} \label{eq_ht_proof_rr_basis}
 \sum_{i=0}^\infty \HasseDer{c}{p^{\mult}-1}(\gamma^{\HTconst+i\HTmult+j} ) X^i \equiv 0 \mod X^{\lenseq-1}, \quad \forall j \in \inter{\noseq+1}
\end{equation}
holds.
Then, the minimum distance of \(\CYC\) is at least $\BoundTwo \defeq \lenseq + \noseq$.
\end{theorem}
\begin{IEEEproof}
Let $c(X) \in \CYC$ and let $\SET{Y} = \{i_0,i_1,\dots,i_{y-1} \}$ denote the support of $c(X)$, where $y \geq d$ holds for all codewords except the all-zero codeword.
We linearly combine the $\noseq+1$ sequences from~\eqref{eq_ht_proof_rr_basis}. Denote the scalar factors for each power series as in~\eqref{eq_ht_proof_rr_basis} by $\lambda_i \in \F{q^{\CYCs}}$ for $i \in \inter{\noseq+1}$. 
We obtain:
\begin{equation} \label{eq_proofHTRR1}
\sum_{i=0}^{\infty} \sum_{j=0}^{\noseq}  \lambda_j \HasseDer{c}{p^{\mult}-1}(\gamma^{\HTconst+i\HTmult +j}) X^i \equiv 0 \mod X^{\lenseq-1}.
\end{equation}
The Hasse derivative (as defined in~\eqref{eq_defHasseDerivative}) of~\eqref{eq_proofHTRR1} leads to:
\begin{equation} \label{eq_proofHTRR2}
\sum_{i=0}^{\infty} \sum_{j=0}^{\noseq}  \lambda_j \left( \sum_{u \in \SET{Y}} \binomsmall{u}{p^{\mult}-1} c_{u}\gamma^{(u-p^{\mult}+1)(\HTconst+i\HTmult+j)} \right) X^i  \equiv 0 \mod X^{\lenseq-1}.
\end{equation}
We re-order~\eqref{eq_proofHTRR2} according to the coefficients of the codeword and obtain:
\begin{align} \label{eq_ProofHTIntermed1}
\sum_{i=0}^{\infty} \sum_{u \in \SET{Y}} \sum_{j=0}^{\noseq} \lambda_j \left( \binomsmall{u}{p^{\mult}-1} c_{u}\gamma^{(u-p^{\mult}+1)(\HTconst+i\HTmult+j)} \right)  X^i & = \sum_{i=0}^{\infty} \sum_{u \in \SET{Y}} \left( \binomsmall{u}{p^{\mult}-1} c_{u}\gamma^{(u-p^{\mult}+1)(\HTconst + im)} \sum_{j=0}^{\noseq} \alpha^{uj}\lambda_j \right) X^i \nonumber \\
&  \equiv 0 \mod X^{\lenseq-1}.
\end{align}
We want to annihilate the first $\noseq$ terms of $c_{i_0}, c_{i_1}, \dots, c_{i_{y-1}}$. From~\eqref{eq_ProofHTIntermed1}, the following linear system of equations with $\noseq+1$ unknowns is obtained:
\begin{align} \label{eq_SystemForCoefficientsa}
\begin{pmatrix}
1 & \gamma^{i_0} & \gamma^{i_0 2} & \cdots & \gamma^{i_0 \noseq}  \\ 
1 & \gamma^{i_1} & \gamma^{i_1 2} & \cdots & \gamma^{i_1 \noseq }  \\ 
\vdots & \vdots  & \vdots & \ddots &  \vdots \\ 
1 & \gamma^{i_{\noseq}} & \gamma^{i_{\noseq} 2} & \cdots & \gamma^{i_{\noseq} \noseq}  \\ 
\end{pmatrix} \cdot
\begin{pmatrix}
\lambda_0 \\ 
\lambda_1 \\ 
\vdots \\
\lambda_{\noseq}
\end{pmatrix}
=
\begin{pmatrix}
0 \\ 
\vdots \\
0 \\ 
1
\end{pmatrix},
\end{align}
and it is guaranteed to find a unique non-zero solution, because the $(\noseq+1) \times (\noseq+1)$ matrix in~\eqref{eq_SystemForCoefficientsa} is a Vandermonde matrix.

Let $\tilde{\SET{Y}} \defeq \SET{Y} \setminus \{ i_0,i_1,\dots, i_{\noseq-1}\}$. Then, we can rewrite~\eqref{eq_ProofHTIntermed1}:
\begin{align*} 
\sum_{i=0}^{\infty} & \left( \sum_{u \in \tilde{\SET{Y}}} \binomsmall{u}{p^{\mult}-1} c_{u} \gamma^{(u-p^{\mult}+1)(\HTconst+i\HTmult)}\sum_{j=0}^{\noseq} \gamma^{uj}\lambda_j \right) X^{i} \equiv 0 \mod X^{\lenseq-1}.
\end{align*}
This leads with the geometric series to:
\begin{align*}
\sum_{u \in \tilde{\SET{Y}}} \frac{  \binomsmall{u}{p^{\mult}-1} c_{u} \gamma^{(u-p^{\mult}+1) \HTconst}
 \sum_{j=0}^{\noseq} \gamma^{uj}\lambda_j }{1-\gamma^{(u-p^{\mult}+1) \HTmult} X} & \equiv 0 \mod X^{\lenseq-1},
\end{align*}
and can be expressed with one common denominator using~\eqref{eq_CommonDenominator} as follows:
\begin{align*}
\frac{\sum\limits_{u \in \tilde{\SET{Y}}} \Bigg(  \binomsmall{u}{p^{\mult}-1} c_{u} \gamma^{(u-p^{\mult}+1) \HTconst}  \sum_{j=0}^{\noseq} \gamma^{uj}\lambda_j      
\frac{\lcm \brac{  1-{\gamma^{(j-p^{\mult}+1)m}X}\ :\ {j \in \SET{Y}}}}{1-{\gamma^{(j-p^{\mult}+1)m}X}} \Bigg)}
{\lcm \brac{  1-{\gamma^{(i-p^{\mult}+1)m}X}\ :\ {i \in \SET{Y}}  }} \equiv 0 \mod X^{\lenseq-1},
\end{align*}
where the degree of the numerator is smaller than or equal to $y-1-\noseq$ and has to be at least $\lenseq-1$. Therefore for $y \geq d$, we have:
\begin{align*}
\CYCd-1-\noseq & \geq \lenseq-1,\\
\CYCd &  \geq \BoundTwo \defeq \lenseq + \noseq.
\end{align*}
\end{IEEEproof}
Note that for $\noseq = 0$, Thm.~\ref{theo_HTRR} becomes Thm.~\ref{theo_BCHRR}.
Thm.~\ref{theo_HTRR} tells us that an $\LINq{n=p^{\mult} \CYCnprime}{k}{d}$ repeated-root cyclic code with generator polynomial $g(X)$ that has $\noseq+1$ sequences of $\delta-1$ consecutive zeros of highest multiplicity $p^{\mult}$, i.e., 
\begin{align*}
\HasseDer{g}{p^{\mult}-1}(\gamma^{\HTconst}) = \HasseDer{g}{p^{\mult}-1}(\gamma^{\HTconst+\HTmult}) =  & \dots = \HasseDer{g}{p^{\mult}-1}(\gamma^{\HTconst +(\lenseq-2)\HTmult}) = 0 \\
\HasseDer{g}{p^{\mult}-1}(\gamma^{\HTconst+1}) = &  \HasseDer{g}{p^{\mult}-1}(\gamma^{\HTconst+\HTmult+1}) = \dots = \HasseDer{g}{p^{\mult}-1}(\gamma^{\HTconst +(\lenseq-2)\HTmult +1}) = 0 \\
& \quad \quad  \quad \quad \vdots \\
& \HasseDer{g}{p^{\mult}-1}(\gamma^{\HTconst + \noseq}) = \HasseDer{g}{p^{\mult}-1}(\gamma^{\HTconst+\HTmult + \noseq}) = \dots= \HasseDer{g}{p^{\mult}-1}(\gamma^{\HTconst +(\lenseq-2)\HTmult +\noseq}) = 0,
\end{align*}
has at least minimum distance $\lenseq+\noseq$. If $s=0$, the repeated-root cyclic code is a simple-root cyclic code and then Thm.~\ref{theo_BCHRR} coincides with the HT bound~\cite{hartmann_decoding_1972, hartmann_generalizations_1972}.
\begin{remark}[Alternative Proof of the Two Bounds] \label{rem_alternative}
An $\LINq{\CYCn=p^{\mult} \CYCnprime}{\CYCk}{\CYCd}$ repeated-root cyclic code with considered consecutive set(s) of zeros with multiplicity $p^{\mult}$ as in Thm.~\ref{theo_BCHRR} and Thm.~\ref{theo_HTRR} is a sub-code of a cyclic code of length $\CYCnprime$ over $\F{q^{p^{\mult}}}$ with the same zeros (see Lemma~\ref{lem_FoldingRepeatedRoots}). 
\end{remark}
Let us consider an example of a binary repeated-root cyclic code and use Thm.~\ref{theo_HTRR} to bound its minimum distance.
\begin{example}[Binary Repeated-Root Cyclic Code] \label{ex_RRCode}
Let $\CYC$ be the binary $\LIN{34 = 2 \cdot 17}{18}{5}{2}$ repeated-root cyclic code with defining set as defined in~\eqref{eq_definingset}:
\begin{equation*}
\defset{\CYC}  = \Big\{ 1^{\langle 2 \rangle}, 2^{\langle 2 \rangle}, 4^{\langle 2 \rangle}, 8^{\langle 2 \rangle}, 9^{\langle 2 \rangle}, 13^{\langle 2 \rangle}, 15^{\langle 2 \rangle}, 16^{\langle 2 \rangle} \Big \},
\end{equation*}
i.e., its generator polynomial is:
\begin{equation*}
g(X) = M_{1,17,2}(X)^2.
\end{equation*}
Thm.~\ref{theo_HTRR} holds for the parameters $\HTconst=1$, $\HTmult=7$, $\lenseq=4$ and $\noseq=1$ and therefore the minimum distance of $\CYC$ is at least 5.
\end{example}

\subsection{Syndrome-Based Burst Error Decoding Algorithm up to Bound I and Bound II} \label{subsec_decodingBoundOneTwo}
Let $\zeta \in \F{q^{p^{\mult}}}$ be such that $(1 \ \zeta \ \dots \ \zeta^{p^{\mult}-1})$ is an $\Fq$-basis of the extension field $\F{q^{p^{\mult}}}$. We define the following bijective map:
\begin{align*}
\phi: \ \F{q}^{p^{\mult}} & \longrightarrow  \F{q^{p^{\mult}}} \\
(a_0 \ a_1 \ \dots \ a_{p^{\mult}-1} ) & \longmapsto a_0 + a_1 \zeta + \cdots a_{p^s-1} \zeta^{p^{\mult}-1}.
\end{align*}
\begin{definition}[Folded Code] \label{def_folding}
Let $\CYC$ be a linear code over $\F{q}$ of length $\CYCn=p^{\mult} \CYCnprime$. The folded code $\CYC^{F}$ of length $\CYCnprime$ over $\F{q^{p^{\mult}}}$ is defined by:
\begin{equation*}
\CYC^{F} \defeq \Big\{ \big(\phi(c_0 \ \dots \ c_{p^{\mult}-1}) \ \dots \
\phi(c_{\CYCn-p^{\mult} } \ \dots \ c_{\CYCn-1}) \big)\ \; | \;  (c_0 \ \dots \ c_{\CYCn-1}) \in \CYC \Big\}.
\end{equation*}
Equivalently, we denote the folded polynomial of a given polynomial $c(X) \in \F{q}[X]$ by $c^{F}(X)$.
\end{definition}
\begin{lemma}[Folding Repeated-Root Cyclic Code] \label{lem_FoldingRepeatedRoots}
Let $\CYC$ be an $\LINq{\CYCn = p^{\mult} \CYCnprime}{\CYCk = p^{\mult} \CYCkprime}{\CYCd}$ repeated-root cyclic code over $\Fq$ with characteristic $p$ and defining set:
\begin{equation*}
\defset{\CYC}  =  \Big\{ i^{\langle p^{\mult} \rangle}  \; | \; i \in  \defset{\CYC^{F}} \Big\},
\end{equation*}
where $|\defset{\CYC^{F}}| = \CYCnprime-\CYCkprime$.
Then the folded code $\CYC^{F}$ as in Def.~\ref{def_folding} is an $\LIN{\CYCnprime}{\CYCkprime}{\CYCd}{q^{p^{\mult}}}$ simple-root cyclic code with defining set $\defset{\CYC^{F}}$.
\end{lemma}
\begin{proof}
Length and dimension of $\CYC^{F}$ follow directly from Def.~\ref{def_folding}. Let us prove the defining set. Every codeword $c(X)$ of the given repeated-root cyclic root $\CYC$ can be written as
\begin{equation*}
c(X) = \sum_{i=0}^{p^{\mult}-1} X^i \sum_{j=0}^{\CYCnprime-1} c_{i+jp^{\mult}}X^{j p^{\mult}} = \sum_{i=0}^{p^{\mult}-1} X^i \sum_{j=0}^{\CYCkprime-1} u_{i,j} X^{j p^{\mult} } g(X^{p^{\mult}}),
\end{equation*}
where $g(X^{p^{\mult}}) = g(X)^{p^{\mult}}$ is the generator polynomial of $\CYC$ with $\CYCnprime-\CYCkprime$ distinct roots of multiplicity $p^{\mult}$.
The corresponding codeword of the folded code $\CYC^{F}$ over $\F{q^{p^{\mult}}}$ in vector notation is:
\begin{equation*}
c^{F}(X) = \sum_{j=0}^{\CYCnprime-1} 
\begin{pmatrix}
c_{0+jp^{\mult}} \\
c_{1+jp^{\mult}} \\
\vdots \\ 
c_{p^{\mult}-1+jp^{\mult}}
\end{pmatrix}
X^{j} = 
\sum_{j=0}^{\CYCkprime-1} 
\begin{pmatrix}
u_{0,j}\\
u_{1,j} \\
\vdots \\ 
u_{p^{\mult}-1,j}
\end{pmatrix}
g(X)
\end{equation*}
and has $\CYCnprime-\CYCkprime$ distinct roots of multiplicity one.
\end{proof}
Folding as given in Def.~\ref{def_folding} is discussed extensively in the literature, especially for Reed--Solomon codes (see e.g.~\cite{krachkovsky_decoding_1997, sidorenko_decoding_2008, guruswami_linear-algebraic_2011}). The operation is essential to decode a given repeated-root cyclic code. 
In the following we describe the decoding approach for $p^{\mult}$-phased burst errors, i.e., errors measured in $\F{q^{p^{\mult}}}$. 
The transmitted (or stored) codeword $c(X)$ of a given $\LINq{p^{\mult} \CYCnprime}{\CYCk}{\CYCd}$ repeated-root cyclic code $\CYC$ is affected by an error $e(X) \in \Fqx$. The received polynomial $r(X) \in \Fqx$ is $r(X) = c(X) + e(X)$. We fold the received word $r(X)$ as in Def.~\ref{def_folding} and obtain
\begin{equation*}
r^{F}(X) = c^{F}(X) + e^{F}(X), 
\end{equation*}
where $e^{F}(X) = \sum_{i \in \ErrorSet} e_i^{F} X^i $ and $\ErrorSet$ is the set of $p^{\mult}$-phased burst error with cardinality $|\ErrorSet| = \NoErrors$.
We describe a syndrome-based decoding procedure up to $ \NoErrors \leq \lfloor (\BoundTwo-1)/2 \rfloor$ $p^{\mult}$-phased burst-errors based on a set of $\noseq+1$ key equations that can be solved by a modified variant of the Extended Euclidean Algorithm (EEA) similar to the procedure to decode simple-root cyclic codes up to the HT bound (see e.g.,~\cite{feng_generalized_1989,feng_generalization_1991,zeh_fast_2011}). Let us first define syndromes in the corresponding extension field. 
\begin{definition}[Syndromes] \label{def_synd}
Let $\CYC$ be an $\LINq{\CYCn}{\CYCk}{\CYCd}$ repeated-root cyclic code over $\Fq$ with characteristic $p$, where $\CYCn = p^{\mult} \CYCnprime$. The integers $\HTconst$, $\HTmult \neq 0$, $\lenseq \geq 2$ and $\noseq \geq 0$ are given as in Thm.~\ref{theo_HTRR}. Let $\gamma \in \F{q^{\CYCs}}$ be an element of order $\CYCnprime$.
We define $\noseq+1$ syndrome polynomials $S^{\langle 0 \rangle}(X), S^{\langle 1 \rangle}(X), \dots, S^{\langle \noseq \rangle}(X)  \in \Fxsub{q^{\CYCs p^{\mult}}}$ for a received polynomial $r(X) \in \Fqx$, respectively the folded version $r^{F}(X) \in \Fxsub{q^{p^{\mult}}}$, as follows:
\begin{equation} \label{eq_defsynd}
S^{\langle t \rangle}(X) \defeq \sum_{i=0}^{\lenseq-2} r^{F}(\gamma^{\HTconst+i\HTmult+t}) X^{i}, \quad \forall t \in \inter{\noseq+1}.
\end{equation}
\end{definition}
To obtain an algebraic description in terms of key equations, we define an error-locator polynomial in the following.
\begin{definition}[Error-Locator Polynomial] \label{def_elp}
Let $\gamma$ be an element of order $\CYCnprime$ in $\F{q^{\CYCs}}$ and let $\HTmult \neq 0$ as in Thm.~\ref{theo_HTRR}. The support of the additive error is $\ErrorSet$ with $|\ErrorSet | = \NoErrors$. Define the error-locator polynomial in $\Fxsub{q^{\CYCs}}$, as:
\begin{equation} \label{eq_defELP}
\Lambda(X) \defeq \prod_{i \in \ErrorSet} \left( 1-X\gamma^{i\HTmult} \right), 
\end{equation}
with degree $\NoErrors$.
\end{definition}
We now connect Def.~\ref{def_synd} and Def.~\ref{def_elp}. From the expression of the syndrome polynomials as in~\eqref{eq_defsynd}, we obtain with the folded received polynomial $r^{F}(X) = c^{F}(X) + e^{F}(X)$:
\begin{align} \label{eq_syndonlyfromerror}
S^{\langle t \rangle}(X) & = \sum_{i=0}^{\lenseq-2} r^{F}(\gamma^{\HTconst+i\HTmult +t }) X^{i} \nonumber \\ 
& = \sum_{i=0}^{\lenseq-2} e^{F}(\gamma^{\HTconst+i\HTmult +t }) X^{i} \nonumber \\ 
& = \sum_{i=0}^{\lenseq-2} \left( \sum_{j \in \ErrorSet} \left( \sum_{u=0}^{p^{\mult}-1} e_{u+jp^{\mult}} \zeta^{u} \right) \gamma^{(\HTconst + i\HTmult + t)j } \right) X^i, \quad \forall t \in \inter{\noseq},
\end{align}
i.e., the syndromes are independent of the folded codeword $c^{F}(X)$. We use the geometric series and we obtain from~\eqref{eq_syndonlyfromerror}:
\begin{align} \label{eq_syndgeometric}
\sum_{i=0}^{\infty} \left( \sum_{j \in \ErrorSet} \left( \sum_{u=0}^{p^{\mult}-1} e_{u+jp^{\mult}} \zeta^{u} \right) \gamma^{(\HTconst + i\HTmult + t)j} \right) X^i \equiv \sum_{j \in \ErrorSet} \left( \sum_{u=0}^{p^{\mult}-1} e_{u+jp^{\mult}} \zeta^{u} \right) \frac{\gamma^{(\HTconst + t) j}}{1-X \gamma^{j\HTmult}} \mod X^{\lenseq-1}.
\end{align}
We need two more steps to obtain a common denominator. From~\eqref{eq_syndgeometric}, we have:
\begin{align} \label{eq_KeyEquation}
S^{\langle t \rangle}(X) & \equiv \sum_{j \in \ErrorSet} 
\left(\sum_{u=0}^{p^{\mult}-1} e_{u+jp^{\mult}} \zeta^{u} \right) \frac{\gamma^{(\HTconst+ t)j }}{1-X \gamma^{j \HTmult}} \mod X^{\lenseq - 1 } \nonumber \\
& \equiv  \frac{\sum_{j \in \ErrorSet} \left( \sum_{u=0}^{p^{\mult}-1} e_{u+jp^{\mult}} \zeta^{u} \right) \gamma^{(\HTconst+ t)j } \prod_{\substack{i \in \ErrorSet \\ i \neq j}} (1-X\gamma^{i \HTmult})}{\prod_{i \in \ErrorSet }(1-X \gamma^{i \HTmult})} \mod X^{\lenseq - 1} \nonumber \\
& \defequiv \frac{\Omega^{\langle t \rangle}(X)}{\Lambda(X)} \mod X^{\lenseq-1}, \quad \forall t \in \inter{\noseq},
\end{align}
where
\begin{equation} \label{eq_eep}
\Omega^{\langle t \rangle}(X) \defeq \sum_{j \in \ErrorSet} \left(\sum_{u=0}^{p^{\mult}-1} e_{u+jp^{\mult}} \zeta^{u} \right) \gamma^{(\HTconst + t) j}  \prod_{\substack{i \in \ErrorSet\\ i \neq j}} \left( 1-X\gamma^{i \HTmult} \right), \quad \forall t \in \inter{\noseq},
\end{equation}
are the $\noseq+1$ error-evaluator polynomials $\Omega^{\langle 0 \rangle}(X),\Omega^{\langle 1 \rangle}(X),\dots, \Omega^{\langle \noseq \rangle}(X)$ of degree at most $\NoErrors -1$.

The $\noseq+1$ key equations as in~\eqref{eq_KeyEquation} can be collaboratively solved by a so-called multisequence shift-register synthesis (see e.g.,~\cite{feng_generalized_1989,feng_generalization_1991}). Algorithm~\ref{algo_EEARR} is based on the Generalized Extended Euclidean Algorithm (GEEA) that solves the corresponding multisequence problem.
\begin{center}
\begin{minipage}[htbp]{.95\textwidth}
\begin{algorithm}[H]
\label{algo_EEARR}
\SetAlgoVlined
\DontPrintSemicolon
\LinesNumbered
\SetKwInput{KwPre}{Preprocessing}
\SetKwInput{KwIn}{Input}
\SetKwInput{KwOut}{Output}
\SetCommentSty{textsf}
\BlankLine
\KwIn{Received word $r(X) \in \Fqx$, element $\gamma$ of order $\CYCnprime$\\ 
\textcolor{white}{\textbf{Input}: }Parameters  $\HTconst$, $\HTmult \neq 0$, $\lenseq \geq 2$ and $\noseq \geq 0$ as in Thm.~\ref{theo_HTRR}} 
\KwOut{Estimated folded codeword $c^{F}(X)$ or \texttt{DecodingFailure}}
\BlankLine 
Calculate $S^{\langle 0 \rangle}(X),\dots, S^{\langle \noseq \rangle}(X)$ as in~\eqref{eq_defsynd} using folded $r^{F}(X)$ \tcp*[r]{Syndrome calculation} 
\BlankLine
$\Lambda(X), \Omega^{\langle 0 \rangle}(X), \dots, \Omega^{\langle \noseq \rangle}(X) = \texttt{GEEA}\big(X^{\lenseq-1}, S^{\langle 0 \rangle}(X), \dots, S^{\langle \noseq \rangle}(X) \big)$ \nllabel{alg_DecodeFrac_CallEEA}\tcp*[r]{Generalized EEA}  
\BlankLine
Find all $i$, where $\Lambda(\gamma_i)=0 \Rightarrow \ErrorSet=\lbrace i_0,i_1,\dots,i_{\NoErrors-1}\rbrace$ \tcp*[r]{Chien-like search} \nllabel{alg_DecodeFrac_FindRoots}
\BlankLine
\eIf{$\NoErrors < \deg \Lambda(X)$}
{
	Declare \texttt{DecodingFailure}\; 
}
{
Determine $e^{F}_{i_0}, e^{F}_{i_1}, \dots, e^{F}_{i_{\NoErrors-1}}$ \nllabel{alg_ErrorEval} \tcp*[r]{Forney error-evaluation}
\BlankLine
$e^{F}(X)$ $\leftarrow$ $\sum_{\ell \in \ErrorSet} e_{\ell}^{F} X^{\ell}$ \;
\BlankLine
$c^F(X)$ $\leftarrow$ $r^{F}(X)- e^{F}(X)$\;
}
\caption{Decoding a $\LIN{p^{\mult} \CYCnprime}{k}{d}{q}$ repeated-root cyclic code $\CYC$ up to $\left \lfloor (\BoundTwo-1)/2 \right \rfloor $ $p^{\mult}$-phased burst errors.}
\end{algorithm}
\end{minipage}
\end{center}
For the $\noseq+2$ input polynomials $X^{\lenseq-1}$ and $S^{\langle 0 \rangle}(X), S^{\langle 1 \rangle}(X), \dots, S^{\langle \noseq \rangle}(X) $ the GEEA returns the polynomials $\Lambda(X)$, $\Omega^{\langle 0 \rangle}(X)$, $\dots$, $\Omega^{\langle \noseq \rangle}(X)$ in $\Fxsub{q^\CYCs}$, such that~\eqref{eq_KeyEquation} holds (as in Line~\ref{alg_DecodeFrac_CallEEA} of Algorithm~\ref{algo_EEARR}). One error-evaluator polynomial $\Omega^{\langle i \rangle}(X)$ as given in~\eqref{eq_eep} is sufficient for the error-evaluation in Line~\ref{alg_ErrorEval}.

Clearly, for $\noseq=0$ Algorithm~\ref{algo_EEARR} decodes up to $\lfloor (\BoundOne-1)/2 \rfloor$  $p^{\mult}$-phased burst errors. Then, the GEEA coincides with the EEA.

\section{Defining Sets of Repeated-Root Cyclic Product Codes} \label{sec_DefiningSet}
Our third lower bound on the minimum distance of a given repeated-root cyclic code $\CYCa$ is based on the embedding of $\CYCa$ into a repeated-root cyclic product code $\CYCa \otimes \CYCb$. Therefore, we explicitly give the defining set of a repeated-root cyclic product code and stress important properties.

Let $\CYCa$ be an $\LINq{\CYCan = p^{\mult} \CYCnprime_a}{\CYCak}{\CYCad}$ repeated-root cyclic code, where $\gcd(\CYCnprime_a,p) = 1$, and let $\CYCb$ be an $\LINq{\CYCbn}{\CYCbk}{\CYCbd}$ simple-root cyclic code. If $\gcd(\CYCan , \CYCbn) = 1$, then the $\LINq{\CYCn = p^{\mult} \CYCnprime_a \CYCbn}{\CYCak \CYCbk}{\CYCad \CYCbd}$ product code $\CYC = \CYCa \otimes \CYCb$ is (repeated-root) cyclic (see e.g.,~\cite[Ch. 18]{macwilliams_theory_1988} for linear product codes). Note that the lengths of two repeated-root cyclic codes over the same field cannot be co-prime and therefore a cyclic product code is not possible. 

Let us investigate the defining set of a repeated-root cyclic product code in the following theorem, originally stated by Burton and Weldon~\cite[Corollary IV]{burton_cyclic_1965}.
\begin{theorem}[Defining Set and Generator Polynomial of a Cyclic Product Code] \label{theo_DefSetProductCode}
Let $\CYCa$ be an $\LINq{\CYCan = p^{\mult} \CYCnprime_a}{\CYCak}{\CYCad}$ repeated-root cyclic code over $\Fq$ with characteristic $p$, and let $\alpha$ be an element of order $\CYCnprime_a$ in $\F{q^{\CYCas}}$. Let $\CYCb$ be an $\LINq{\CYCbn}{\CYCbk}{\CYCbd}$ simple-root cyclic code and let $\beta$ be an element of order $\CYCbn$ in $\F{q^{\CYCbs}}$. Let $\CYCs = \lcm(\CYCas,\CYCbs)$. The defining sets of $\CYCa$ and $\CYCb$ are denoted by $\defset{\CYCa}$ respectively $\defset{\CYCb}$ and their generator polynomials by $g_a(X)$ respectively $g_b(X)$. Let two integers $\inta$ and $\intb$ be given, such that:
\begin{equation*}
\inta \CYCan + \intb \CYCbn =1.
\end{equation*}
The generator polynomial $g(X)$ of the repeated-root cyclic product code $\CYCa \otimes \CYCb$ is:
\begin{equation} \label{eq_GenPolyCyclicProduct}
g(X) = \gcd \Big( X^{\CYCan \CYCbn}-1, g_a(X^{\intb \CYCbn}) \cdot g_b(X^{\inta \CYCan}) \Big).
\end{equation}
Let $\gamma \defeq \alpha \beta$ in $\F{q^{\CYCs}}$ and let:
\begin{equation} \label{eq_DefSetSimpleRoot}
\bar{\mathsf{D}}_{\CYCb} = \Big\{ i^{\langle p^{\mult}  \rangle} \; | \; i \in \defset{\CYCb} \Big\}.
\end{equation}
Then the defining set of the repeated-root cyclic product code $\CYC = \CYCa \otimes \CYCb$ is:
\begin{align*} \label{eq_DefSetCyclicProduct}
\defset{\CYC} =  \Big\{ \defset{\CYCa} \cup \defset[\CYCnprime_a]{\CYCa} \cup \defset[2 \CYCnprime_a]{\CYCa} \cup \dots \cup \defset[(\CYCbn-1)\CYCnprime_a]{\CYCa} \Big\} \cupmax \Big\{ 
\bar{\mathsf{D}}_{\CYCb} \cup \bar{\mathsf{D}}_{\CYCb}^{[\CYCbn]} \cup \bar{\mathsf{D}}_{\CYCb}^{[2\CYCbn]} \cup \dots \cup \bar{\mathsf{D}}_{\CYCb}^{[(\CYCnprime_a-1)\CYCbn]} \Big\},
\end{align*}
where $\defset[\CYCnprime_a]{\CYCa}$ was defined in~\eqref{eq_definingsetext} and the operation in~\eqref{eq_cupmaxdefintion}.
\end{theorem}
For the proof we refer to the proof of~\cite[Thm. 3 and Corollary IV]{burton_cyclic_1965}. We explicitly give the defining set of the repeated-root cyclic product code $\CYCa \otimes \CYCb$ here and we want to emphasize that the roots of the simple-root cyclic code $\CYCb$ have highest multiplicity $p^{\mult}$ in the defining set of $\CYCa \otimes \CYCb$ (see~\eqref{eq_DefSetSimpleRoot}), because 
\begin{equation*}
g_b(X^{\inta \CYCan}) = g_b(X^{\inta \CYCnprime_a})^{p^{\mult}}.
\end{equation*}

\section{Bound III: Embedding into Repeated-Root Cyclic Product Codes} \label{sec_BoundsProduct}
Similar to Thm. 4 of~\cite{zeh_generalizing_2013} for a simple-root cyclic code, we embed a given repeated-root cyclic code $\CYCa$ into a repeated-root cyclic product code $\CYCa \otimes \CYCb$ to bound the minimum distance of $\CYCa$.
\begin{theorem}[Bound III: Embedding into a Product Code] \label{theo_GenBCHboundRR}
Let $\CYCa$ be an $\LINq{\CYCan = p^{\mult} \CYCnprime_a}{\CYCak}{\CYCad}$ repeated-root cyclic code over $\Fq$ with characteristic $p$, where $\gcd(\CYCnprime_a,p)=1$ and let $\CYCb$ be an $\LINq{\CYCbn}{\CYCbk}{\CYCbd}$ simple-root cyclic code, respectively, with $\gcdab{\CYCan}{\CYCbn}=1$.
Let $\alpha$ be an element of order $\CYCnprime_a$ in $\F{q^{\CYCas}}$, $\beta$ of order $\CYCbn$ in $\F{q^{\CYCbs}}$, respectively, and let two integers $\HTconsta$, $\HTconstb$ and two non-zero integers $\HTmulta \neq 0$, $\HTmultb \neq 0$ with $\gcdab{\CYCan}{\HTmulta}=\gcdab{\CYCbn}{\HTmultb}=1$ be given. Assume that for all codewords $a(X) \in \CYCa$ and  $b(X) \in \CYCb$ 
\begin{equation} \label{eq_statementBCHgen}
\sum_{i=0}^{\infty}  \HasseDer{a}{p^{\mult}-1}(\alpha^{\HTconsta+i\HTmulta}) \cdot b(\beta^{\HTconstb+i\HTmultb}) X^{i} \equiv 0 \mod X^{\lenseq-1}
\end{equation}
holds for some integer $\lenseq \geq 2$.
Then, we obtain:
\begin{equation} 
\CYCad \geq \BoundThree \defeq \left \lceil \frac{\lenseq}{\CYCbd} \right \rceil.
\end{equation}
\begin{IEEEproof}
From Thm.~\ref{theo_DefSetProductCode} we know that \eqref{eq_statementBCHgen} corresponds to $\lenseq-1$ consecutive zeros with highest multiplicity $p^{\mult}$ of the repeated-root cyclic product code $\CYCa \otimes \CYCb$. By Thm.~\ref{theo_BCHRR}, the minimum distance $\CYCd$ of $\CYCa \otimes \CYCb$ is greater than or equal to $\lenseq$. Therefore:
\begin{align*}
\CYCd = \CYCad \CYCbd \geq \lenseq \quad \Longleftrightarrow \quad \CYCad = \left \lceil \frac{\lenseq}{\CYCbd} \right \rceil.
\end{align*}
\end{IEEEproof}
\end{theorem}
Note that the expression of~\eqref{eq_statementBCHgen} is in $\Fxsub{q^\CYCs}$, where $\CYCs = \lcmab{\CYCas}{\CYCbs}$.
\begin{example}[Bound by Embedding into a Product Code]
Let $\CYCa$ be the $\LIN{34 = 2 \cdot 17}{18}{5}{2}$ repeated-root cyclic code with $p=2$, $\CYCnprime_a = 17$ and defining set:
\begin{equation*}
\defset{\CYCa}  = \Big\{ 1^{\langle 2 \rangle}, 2^{\langle 2 \rangle}, 4^{\langle 2 \rangle}, 8^{\langle 2 \rangle}, 9^{\langle 2 \rangle}, 13^{\langle 2 \rangle}, 15^{\langle 2 \rangle}, 16^{\langle 2 \rangle} \Big\},
\end{equation*}
of Ex.~\ref{ex_RRCode} and let $\CYCb$ denote the $\LIN{3}{2}{2}{2}$ simple-root cyclic parity check code with defining set
\begin{equation*}
{\defset{\CYCb}} = \Big\{ 0^{\langle 1 \rangle} \Big\}.
\end{equation*}
Let $\alpha \in \F{2^8}$ and $\beta \in \F{2^4}$ denote elements of order 17 and 3, respectively. Then, for $\HTconsta=-4$, $\HTconstb=-1$ and $\HTmulta=\HTmultb=1$ Thm.~\ref{theo_GenBCHboundRR} holds for $\delta=10$ and therefore $\CYCad \geq 5$, which is the true minimum distance of $\CYCa$. 

Since $1 \cdot 34  - 11 \cdot 3 = 1$, according to Thm.~\ref{theo_DefSetProductCode}, the defining set of the repeated-root cyclic product code $\CYCa \otimes \CYCb$ is:
\begin{align*}
 \defset{\CYCa \otimes \CYCb}= & \Big\{ \big\{ 1^{\langle 2 \rangle}, 2^{\langle 2 \rangle}, 4^{\langle 2 \rangle}, 8^{\langle 2 \rangle}, 9^{\langle 2 \rangle}, 13^{\langle 2 \rangle}, 15^{\langle 2 \rangle}, 16^{\langle 2 \rangle} \big\} \cup \big\{ 18^{\langle 2 \rangle}, 19^{\langle 2 \rangle}, 21^{\langle 2 \rangle}, 25^{\langle 2 \rangle},26^{\langle 2 \rangle}, 30^{\langle 2 \rangle}, 32^{\langle 2 \rangle}, 33^{\langle 2 \rangle} \big\} \\ 
& \quad \cup \big\{ 35^{\langle 2 \rangle}, 36^{\langle 2 \rangle}, 38^{\langle 2 \rangle}, 42^{\langle 2 \rangle},43^{\langle 2 \rangle}, 47^{\langle 2 \rangle},49^{\langle 2 \rangle},50^{\langle 2 \rangle} \big\} \Big\} \cupmax \Big\{ \big\{0^{\langle 2 \rangle}\big\} \cup \big\{3^{\langle 2 \rangle}\big\} \cup \cdots \cup \big\{48^{\langle 2 \rangle}\big\} \Big\} \\
= & \Big\{ 0^{\langle 2 \rangle},1^{\langle 2 \rangle},2^{\langle 2 \rangle},3^{\langle 2 \rangle},4^{\langle 2 \rangle},6^{\langle 2 \rangle},8^{\langle 2 \rangle},9^{\langle 2 \rangle},12^{\langle 2 \rangle},13^{\langle 2 \rangle},15^{\langle 2 \rangle},16^{\langle 2 \rangle},18^{\langle 2 \rangle},19^{\langle 2 \rangle},21^{\langle 2 \rangle},24^{\langle 2 \rangle},25^{\langle 2 \rangle},26^{\langle 2 \rangle},27^{\langle 2 \rangle},\\
& \quad 30^{\langle 2 \rangle},32^{\langle 2 \rangle},33^{\langle 2 \rangle},35^{\langle 2 \rangle},36^{\langle 2 \rangle},38^{\langle 2 \rangle},39^{\langle 2 \rangle},42^{\langle 2 \rangle},43^{\langle 2 \rangle},45^{\langle 2 \rangle},47^{\langle 2 \rangle},48^{\langle 2 \rangle},49^{\langle 2 \rangle},50^{\langle 2 \rangle} \Big\}.
\end{align*}
\end{example}

\section{Probabilistic Decoding up to Bound III}  \label{sec_Decoding}
In contrast to the decoding approach for $p^{\mult}$-phased burst errors in Section~\ref{subsec_decodingBoundOneTwo}, we do not use folding (as in Def.~\ref{def_folding}) in the following. Instead we decode a given $\LINq{\CYCan = p^{\mult} \CYCnprime_a}{\CYCak}{\CYCad}$ repeated-root cyclic code $\CYCa$ (embedded in a repeated-root cyclic product code $\CYCa \otimes \CYCb$ via an associated single-root cyclic code $\CYCb$ as in Thm.~\ref{theo_GenBCHboundRR}) as a $p^{\mult}$-interleaved code and apply a probabilistic decoder (as e.g. analyzed in~\cite{krachkovsky_decoding_1997, krachkovsky_decoding_1998, schmidt_collaborative_2009}). Note that this decoding method also corrects $p^{\mult}$-phased burst errors.
Let $a(X) \in \CYCa$ and let the received polynomial be $r(X) = a(X) + e(X)$.

Let $p^{\mult}$ polynomials $r^{\langle 0 \rangle}(X), r^{\langle 1 \rangle}(X), \dots, r^{\langle p^{\mult}-1 \rangle}(X) \in \Fqx$ of degree smaller than $\CYCnprime_a$ be given, such that
\begin{equation} \label{eq_RecWords}
r(X) = \sum_{i=0}^{p^s-1} r^{\langle i \rangle}(X^{p^{\mult}}) X^{i}, 
\end{equation} 
where $\ErrorSet_i$ denotes the corresponding error-positions in $r^{\langle i \rangle}(X)$. The set $\ErrorSet = \cup_{i=0}^{p^{\mult}-1} \ErrorSet_i $ with $\NoErrors = |\ErrorSet|$ is the set of $p^{\mult}$-phased burst-errors.

In the following the set of $p^{\mult}$ key equations is derived and the decoding procedure up to  $ \NoErrors \leq \lfloor \frac{p^{\mult}}{p^{\mult}-1}(\BoundThree-1) \rfloor $ $p^{\mult}$-phased burst errors is described.
\begin{definition}[Syndromes] \label{def_syndb}
Let $\CYCa$ be an $\LINq{\CYCan = p^{\mult} \CYCnprime_a}{\CYCak}{\CYCad}$ repeated-root cyclic code over $\Fq$ with characteristic $p$, where $\gcd(\CYCnprime_a,p)=1$, and $\CYCb$ an $\LINq{\CYCbn}{\CYCbk}{\CYCbd}$ simple-root cyclic code, respectively, with $\gcdab{\CYCan}{\CYCbn}=1$. Let $\alpha$, $\beta$ be elements of order $\CYCnprime_a$ in $\F{q^{\CYCas}}$ and of order $\CYCbn$ in $\F{q^{\CYCbs}}$ respectively. The integers $\HTconsta$, $\HTconstb$, $\HTmulta \neq 0$, $\HTmultb \neq 0$ with $\gcdab{\CYCan}{\HTmulta}=\gcdab{\CYCbn}{\HTmultb}=1$ and $\lenseq \geq 2$ are given as in Thm.~\ref{theo_GenBCHboundRR}.
Furthermore, let $b(X) \in \CYCb$ be a codeword of weight $\CYCbd$. We define $p^{\mult}$ syndrome polynomials $S^{\langle 0  \rangle}(X), S^{\langle 1 \rangle}(X), \dots, S^{\langle p^{\mult-1} \rangle}(X) \in \Fxsub{q^\CYCs}$, where $\CYCs = \lcm(\CYCas, \CYCbs)$ for the received polynomials $r^{\langle 0 \rangle}(X), r^{\langle 1 \rangle}(X), \dots, r^{\langle p^{\mult}-1 \rangle}(X) \in \Fqx$ as in~\eqref{eq_RecWords}:
\begin{equation} \label{eq_defsyndb}
S^{\langle t \rangle}(X) \defeq \sum_{i=0}^{\lenseq-2} r^{\langle t \rangle}(\alpha^{\HTconsta+i\HTmulta}) \cdot b(\beta^{\HTconstb+i\HTmultb}) X^{i}, \quad \forall t \in \inter{p^{\mult}}.
\end{equation}
\end{definition}
To obtain an algebraic description in terms of a key equation, we define an error-locator polynomial in the following.
\begin{definition}[Error-Locator Polynomial] \label{def_elpb}
Let $b(X) = \sum_{i \in \SupportSecond} b_i X^i $ be a codeword of weight $|\SupportSecond| = \CYCbd$ of the associated $\LINq{\CYCbn}{\CYCbk}{\CYCbd}$ simple-root cyclic code $\CYCb$. Let $\alpha$ and $\beta$ be elements of order $\CYCnprime_a$ in $\F{q^{\CYCas}}$ and of order $\CYCbn$ in $\F{q^{\CYCbs}}$, respectively, and let $\HTmulta \neq 0$ and $\HTmultb \neq 0$ be as in Thm.~\ref{theo_GenBCHboundRR}.

The support of the additive error is $\ErrorSet$ with $|\ErrorSet | = \NoErrors$. Define the error-locator polynomial in $\Fxsub{q^{\CYCs}}$, where $\CYCs = \lcm(\CYCas, \CYCbs)$, as:
\begin{equation} \label{eq_defelpb}
\Lambda(X) \defeq \prod_{i \in \ErrorSet} \left( \prod_{j \in \SupportSecond} \left( 1-X\alpha^{i\HTmulta} \beta^{j\HTmultb} \right) \right), 
\end{equation}
with degree $\NoErrors \cdot \CYCbd$.
\end{definition}
For some $j \in \SupportSecond$, let $\CYCnprime_a$ distinct roots of the error-locator polynomial $\Lambda(X)$, as defined in~\eqref{eq_defelpb}, be denoted as:
\begin{equation} \label{eq_RootOfELP}
\gamma_i \defeq \beta^{-j \HTmultb} \alpha^{-i \HTmulta}, \quad \forall i \in \inter{\CYCnprime_a}.
\end{equation}
We pre-calculate $\CYCnprime_a$ roots as in~\eqref{eq_RootOfELP} and identify the error positions of a given error-locator polynomial $\Lambda(X)$ as in Def.~\ref{def_elpb}.

We now connect Def.~\ref{def_syndb} and Def.~\ref{def_elpb}. From the expression of the syndromes in~\eqref{eq_defsyndb}, we obtain:
\begin{align} \label{eq_syndonlyfromerrorb}
S^{\langle t \rangle}(X) & = \sum_{i=0}^{\lenseq-2} r^{\langle t \rangle}(\alpha^{\HTconsta+i\HTmulta}) \cdot b(\beta^{\HTconstb+i\HTmultb}) X^{i} \nonumber \\
& = \sum_{i=0}^{\lenseq-2} e^{\langle t \rangle}(\alpha^{\HTconsta+i\HTmulta}) \cdot b(\beta^{\HTconstb+i\HTmultb}) X^{i} \nonumber \\ 
& = \sum_{i=0}^{\lenseq-2} \left( \sum_{j \in \ErrorSet_t}  e^{\langle t \rangle}_j \alpha^{(\HTconsta + i\HTmulta)j} \cdot \sum_{l \in \SupportSecond} b_l \beta^{(\HTconstb+i\HTmultb)l} \right) X^i, \quad \forall t \in \inter{p^{\mult}}.
\end{align}
As in~\eqref{eq_syndgeometric}, we use the geometric series and we obtain from~\eqref{eq_syndonlyfromerrorb}:
\begin{align} \label{eq_syndgeometricb}
\sum_{i=0}^{\infty} \left( \sum_{j \in \ErrorSet_t} e^{\langle t \rangle}_j \alpha^{(\HTconsta + i\HTmulta)j}  \cdot \sum_{l \in \SupportSecond} b_l \beta^{(\HTconstb+i\HTmultb)l} \right) X^i \equiv \sum_{j \in \ErrorSet_t} e^{\langle t \rangle}_j \alpha^{\HTconsta j}  \sum_{l \in \SupportSecond} \frac{b_l \beta^{\HTconstb l}}{1-X \alpha^{j \HTmulta} \beta^{l \HTmultb} } \mod X^{\lenseq - 1 }.
\end{align}
We need two more steps to obtain a common denominator. From~\eqref{eq_syndgeometricb}, we have:
\begin{align} 
S^{\langle t \rangle}(X) & \equiv \sum_{j \in \ErrorSet_t} e^{\langle t \rangle}_j \alpha^{\HTconsta j}  \sum_{l \in \SupportSecond} \frac{b_l \beta^{\HTconstb l}}{1-X \alpha^{j \HTmulta} \beta^{l \HTmultb} } \mod X^{\lenseq - 1 } \nonumber \\
&  \equiv \sum_{j \in \ErrorSet_t} e^{\langle t \rangle}_j \alpha^{\HTconsta j} \frac{\sum_{l \in \SupportSecond} b_l \beta^{\HTconstb l} \prod_{\substack{i \in \SupportSecond \\ i \neq l}} (1-X\alpha^{j \HTmulta} \beta^{i \HTmultb}) }{\prod_{i \in \SupportSecond}(1-X \alpha^{j \HTmulta} \beta^{i \HTmultb})} \mod X^{\lenseq - 1} \nonumber \\
&  \equiv \frac{\sum_{j \in \ErrorSet_t} \left(  e^{\langle t \rangle}_j \alpha^{\HTconsta j} \sum_{l \in \SupportSecond} \left( b_l \beta^{\HTconstb l} \prod_{\substack{i \in \SupportSecond \\ i \neq l}} (1-X\alpha^{j \HTmulta} \beta^{i \HTmultb}) \right) \prod_{\substack{s \in \ErrorSet\\ s \neq j}} \prod_{\iota \in \SupportSecond} \left( 1-X\alpha^{s \HTmulta}\beta^{\iota \HTmultb} \right) \right)}{\prod_{i \in \ErrorSet_t} \left( \prod_{j \in \SupportSecond} \left( 1-X\alpha^{i\HTmulta} \beta^{j\HTmultb} \right) \right)} \nonumber \\
& \defequiv \frac{\Omega^{\langle t \rangle}(X)}{\Lambda(X)} \mod X^{\lenseq-1}, \quad \forall t \in \inter{p^{\mult}}, \label{eq_KeyEquationb}
\end{align}
where 
\begin{equation} \label{eq_eepb}
\Omega^{\langle t \rangle}(X) \defeq \sum_{j \in \ErrorSet_t} \left(  e^{\langle t \rangle}_j \alpha^{\HTconsta j} \sum_{l \in \SupportSecond} \left( b_l \beta^{\HTconstb l} \prod_{\substack{i \in \SupportSecond \\ i \neq l}} (1-X\alpha^{j \HTmulta} \beta^{i \HTmultb}) \right) \prod_{\substack{s \in \ErrorSet\\ s \neq j}} \prod_{\iota \in \SupportSecond} \left( 1-X\alpha^{s \HTmulta}\beta^{\iota \HTmultb} \right) \right), \quad \forall t \in \inter{p^{\mult}}
\end{equation}
are the $p^{\mult}$ error-evaluator polynomials $\Omega^{\langle 0 \rangle}(X), \Omega^{\langle 1 \rangle}(X), \dots, \Omega^{\langle p^{\mult}-1 \rangle}(X)$ of degree at most $\NoErrors \CYCbd -1$.
We skip the explicit error-evaluation and refer to~\cite[Proposition 4]{zeh_new_2012}.

Algorithm~\ref{algo_EEARRb} summarizes the syndrome-based decoding procedure up to $\left \lfloor \frac{p^{\mult}}{p^{\mult}-1}(\BoundThree-1) \right \rfloor $ $p^{\mult}$-phased burst errors with high probability based on the key equations as in~\eqref{eq_KeyEquationb} and ~\eqref{eq_eepb}. 
\begin{center}
\begin{minipage}[htbp]{.95\textwidth}
\begin{algorithm}[H]
\label{algo_EEARRb}
\SetAlgoVlined
\DontPrintSemicolon
\LinesNumbered
\SetKwInput{KwPre}{Preprocess}
\SetKwInput{KwIn}{Input}
\SetKwInput{KwOut}{Output}
\SetCommentSty{textsf}
\BlankLine
\KwIn{Received word $r(X)$, codeword $b(X) = \sum_{i \in \SupportSecond} b_i X^i \in \CYCb$,\\
\textcolor{white}{\textbf{Input}: }Elements $\alpha$ and $\beta$ of order $\CYCan$ and $\CYCbn$,\\
\textcolor{white}{\textbf{Input}: }Parameters $\HTconsta$, $\HTconstb$, $\HTmulta \neq 0$, $\HTmultb \neq 0$ and $\lenseq \geq 2$ as in Thm.~\ref{theo_GenBCHboundRR}} 
\KwOut{Estimated codeword $a(X) \in \CYCa$ or \texttt{DecodingFailure}}
\BlankLine 
\KwPre{\\
\hspace{1em}\textbf{for} all $i \in \inter{\CYCnprime_a}$: calculate $\gamma_i = \beta^{-j\HTmultb} \alpha^{-i\HTmulta} $, where $j \in \SupportSecond$
}
\BlankLine 
Calculate $S^{\langle 0 \rangle}(X), \dots, S^{\langle p^{\mult}-1 \rangle}(X)$ as in~\eqref{eq_defsyndb} using $r^{\langle 0 \rangle}(X), \dots, r^{\langle p^{\mult}-1 \rangle}(X) $ \tcp*[r]{Syndrome calculation} 
\BlankLine
$\Lambda(X), \Omega^{\langle 0 \rangle}(X), \dots, \Omega^{\langle p^{\mult}-1 \rangle}(X)  = \texttt{GEEA}\big(X^{\lenseq-1}, S^{\langle 0 \rangle}(X),\dots, S^{\langle p^{\mult}-1  \rangle}(X)  \big)$ \nllabel{alg_DecodeFrac_CallGEEAb}\tcp*[r]{Generalized EEA}  
\BlankLine
Find all $i$, where $\Lambda(\gamma_i)=0 \Rightarrow \ErrorSet=\lbrace i_0,i_1,\dots,i_{\NoErrors-1}\rbrace$ \tcp*[r]{Chien-like search} \nllabel{alg_DecodeFrac_FindRootsb}
\BlankLine
\eIf{$\NoErrors | \SupportSecond | < \deg \Lambda(X)$}
{
	Declare \texttt{DecodingFailure}\; 
}
{
\textbf{for} all $i \in \inter{p^{\mult}}$: Determine $e^{\langle i \rangle}(X)$ using $\Omega^{\langle i \rangle}(X)$ as in~\cite[Proposition 4]{zeh_new_2012} \nllabel{alg_ErrorEvalb} \tcp*[r]{Forney-like error-evaluation}
\textcolor{white}{\textbf{for} all $i \in \{p^{\mult}\}$:} $e^{\langle i \rangle}(X)$ $\leftarrow$ $\sum_{j \in \ErrorSet_i} e_{j}^{\langle i \rangle}  X^{j}$ \;
\BlankLine
$a(X)$ $\leftarrow$ $\sum_{i=0}^{p^s-1} \left( r^{\langle i \rangle}(X^{p^{\mult}}) - e^{\langle i \rangle}(X^{p^{\mult}}) \right) X^{i}$\;
}
\caption{Decoding a $\LIN{p^{\mult} \CYCnprime_a}{\CYCak}{\CYCad}{q}$ repeated-root cyclic code $\CYCa$ up to $\left \lfloor \frac{p^{\mult}}{p^{\mult}-1}(\BoundThree-1) \right \rfloor $ $p^{\mult}$-phased burst errors.}
\end{algorithm}
\end{minipage}
\end{center}
All error-evaluator polynomials $\Omega^{\langle 0 \rangle}(X), \Omega^{\langle 1 \rangle}(X), \dots, \Omega^{\langle p^{\mult}-1 \rangle}(X)$ as defined in~\eqref{eq_eep} are needed for the error-evaluation in Line~\ref{alg_ErrorEval}.

Bound III simplifies to the BCH-like generalization of Bound I (as stated in Thm.~\ref{theo_BCHRR}) if the associated code $\CYCb$ is the trivial $\LINq{\CYCbn}{\CYCbn}{1}$ code and decoding up to $\frac{p^{\mult}-1}{p^{\mult}} \lfloor (\BoundOne-1)/2 \rfloor $ $p^{\mult}$-phased burst errors with high probability is possible. Then the $p^s$ parallel operations (as e.g., the syndrome calculation) are computed over $\F{q^{\CYCs}}$ instead in $\F{q^{\CYCs p^{\mult}}}$.

Note that the $p^{\mult}$ cyclic subcodes can be collaboratively list decoded with the approach of Gopalan~\cite{gopalan_list_2011} up to the $q$-ary Johnson radius with relative distance $\BoundOne/\CYCan$.

\section{Conclusion} \label{sec_Conclusion}
We have proved three lower bounds on the minimum distance of a repeated-root cyclic code, i.e., a cyclic code whose length is not relatively prime to the field characteristic. The two first bounds are generalizations of the BCH and the HT bound to repeated-root cyclic codes. A syndrome-based decoding algorithm with a guaranteed radius was developed. The third bound is similar to a previous published technique for simple-root cyclic codes and is based on the embedding of a given repeated-root cyclic code into a repeated-root cyclic product code.
A syndrome-based probabilistic decoding algorithm based on a set of key equations using the third bound was proposed.

\section*{Acknowledgments}
The authors are grateful to Antonia Wachter-Zeh, Johan S.~R. Nielsen and Ron M. Roth for stimulating discussions.


\begin{thebibliography}{10}
\providecommand{\url}[1]{#1}
\csname url@samestyle\endcsname
\providecommand{\newblock}{\relax}
\providecommand{\bibinfo}[2]{#2}
\providecommand{\BIBentrySTDinterwordspacing}{\spaceskip=0pt\relax}
\providecommand{\BIBentryALTinterwordstretchfactor}{4}
\providecommand{\BIBentryALTinterwordspacing}{\spaceskip=\fontdimen2\font plus
\BIBentryALTinterwordstretchfactor\fontdimen3\font minus
  \fontdimen4\font\relax}
\providecommand{\BIBforeignlanguage}[2]{{%
\expandafter\ifx\csname l@#1\endcsname\relax
\typeout{** WARNING: IEEEtran.bst: No hyphenation pattern has been}%
\typeout{** loaded for the language `#1'. Using the pattern for}%
\typeout{** the default language instead.}%
\else
\language=\csname l@#1\endcsname
\fi
#2}}
\providecommand{\BIBdecl}{\relax}
\BIBdecl

\bibitem{bose_class_1960}
R.~C. Bose and D.~K. Ray-Chaudhuri, ``{On A Class of Error Correcting Binary
  Group Codes},'' \emph{Inf. Control}, vol.~3, no.~1, pp. 68--79, 1960.

\bibitem{hocquenghem_codes_1959}
A.~Hocquenghem, ``Codes correcteurs {d'Erreurs},'' \emph{Chiffres (Paris)},
  vol.~2, pp. 147--156, 1959.

\bibitem{hartmann_decoding_1972}
C.~R.~P. Hartmann, ``{Decoding Beyond the {BCH} Bound},'' \emph{IEEE Trans.
  Inform. Theory}, vol.~18, no.~3, pp. 441--444, 1972.

\bibitem{hartmann_generalizations_1972}
C.~R.~P. Hartmann and K.~K. Tzeng, ``{Generalizations of the {BCH} Bound},''
  \emph{Inf. Control}, vol.~20, no.~5, pp. 489--498, 1972.

\bibitem{berman_semisimple_1967}
S.~D. Berman, ``\BIBforeignlanguage{en}{{Semisimple Cyclic and Abelian
  Codes}},'' \emph{\BIBforeignlanguage{en}{Cybernetics}}, vol.~3, no.~3, pp.
  17--23, 1967.

\bibitem{massey_polynomial_1973}
J.~L. Massey, D.~Costello, and J.~Justesen, ``{Polynomial Weights and Code
  Constructions},'' \emph{IEEE Trans. Inform. Theory}, vol.~19, no.~1, pp.
  101--110, 1973.

\bibitem{massey_hasse_1986}
J.~L. Massey, N.~von Seemann, and P.~Schöller, ``{Hasse Derivatives and
  Repeated-Root Cyclic Codes},'' in \emph{{IEEE} International Symposium on
  Information Theory ({ISIT)}}, 1986, p.~39.

\bibitem{da_rocha_jr._repeated-single-root_1994}
V.~C. da~Rocha~Jr., ``\BIBforeignlanguage{en}{{On Repeated-Single-Root
  Constacyclic Codes}},'' in \emph{\BIBforeignlanguage{en}{Communications and
  Cryptography}}, ser. The Springer International Series in Engineering and
  Computer Science.\hskip 1em plus 0.5em minus 0.4em\relax Springer {US}, 1994,
  no. 276, pp. 93--99.

\bibitem{castagnoli_minimum_1989}
G.~Castagnoli, ``{On the Minimum Distance of Long Cyclic Codes and Cyclic
  Redundancy Check Codes},'' in \emph{PhD Thesis, {ETH} Zürich}, 1989.

\bibitem{castagnoli_asymptotic_1989}
------, ``{On the Asymptotic Badness of Cyclic Codes with Block-Lengths
  Composed from a Fixed Set of Prime Factors},'' in \emph{Applied Algebra,
  Algebraic Algorithms and Error-Correcting Codes}, vol. 357.\hskip 1em plus
  0.5em minus 0.4em\relax Springer {Berlin/Heidelberg}, 1989, pp. 164--168.

\bibitem{castagnoli_repeated-root_1991}
G.~Castagnoli, J.~L. Massey, P.~A. Schöller, and N.~von Seemann, ``{On
  Repeated-Root Cyclic Codes},'' \emph{IEEE Trans. Inform. Theory}, vol.~37,
  no.~2, pp. 337--342, 1991.

\bibitem{latypov_checking_1988}
R.~K. Latypov, ``{Checking Matrix of a Cyclic Code Generated by Multiple
  Roots},'' \emph{Journal of Soviet Mathematics}, vol.~43, no.~3, pp.
  2492--2495, 1988.

\bibitem{van_lint_repeated-root_1991}
J.~van Lint, ``{Repeated-Root Cyclic Codes},'' \emph{IEEE Trans. Inform.
  Theory}, vol.~37, no.~2, pp. 343--345, 1991.

\bibitem{zimmermann_generalizations_1996}
K.-H. Zimmermann, ``{On Generalizations of Repeated-Root Cyclic Codes},''
  \emph{IEEE Trans. Inform. Theory}, vol.~42, no.~2, pp. 641--649, 1996.

\bibitem{nedeloaia_weight_2003}
C.-S. Nedeloaia, ``{Weight Distributions of Cyclic Self-Dual Codes},''
  \emph{IEEE Trans. Inform. Theory}, vol.~49, no.~6, pp. 1582--1591, 2003.

\bibitem{ling_algebraic_2006}
S.~Ling, H.~Niederreiter, and P.~Solé, ``{On the Algebraic Structure of
  Quasi-cyclic Codes {IV:} Repeated Roots},'' \emph{Des. Codes Cryptogr.},
  vol.~38, no.~3, pp. 337--361, 2006.

\bibitem{dinh_repeated-root_2012}
H.~Q. Dinh, ``{Repeated-Root Constacyclic Codes of Length $2p^s$},''
  \emph{Finite Fields Th. App.}, vol.~18, no.~1, pp. 133--143, 2012.

\bibitem{dinh_structure_2013}
------, ``{Structure of Repeated-Root Constacyclic Codes of Length $3p^s$ and
  Their Duals},'' \emph{Discrete Math.}, vol. 313, no.~9, pp. 983--991, 2013.


S.~Ling, H.~Niederreiter, and P.~Solé, ``{On the Algebraic Structure of
  Quasi-cyclic Codes {IV:} Repeated Roots},'' \emph{Des. Codes Cryptogr.},
  vol.~38, no.~3, pp. 337--361, 2006.

\bibitem{zeh_new_2012}
A.~Zeh and S.~V. Bezzateev, ``{A New Bound on the Minimum Distance of Cyclic 
Codes Using Small-Minimum-Distance Cyclic Codes},'' \emph{Des. Codes Cryptogr.}, vol.~71, no.~2, pp. 229--246, 2014.


\bibitem{zeh_decoding_2012}
A.~Zeh, A.~Wachter-Zeh, and S.~V. Bezzateev, ``{Decoding Cyclic Codes up to a
  New Bound on the Minimum Distance},'' \emph{IEEE Trans. Inform. Theory},
  vol.~58, no.~6, pp. 3951--3960, 2012.

\bibitem{burton_cyclic_1965}
H.~Burton and E.~J. Weldon, ``{Cyclic Product Codes},'' \emph{IEEE Trans.
  Inform. Theory}, vol.~11, no.~3, pp. 433--439, 1965.

\bibitem{lin_further_1970}
S.~Lin and E.~J. Weldon, ``{Further Results on Cyclic Product Codes},''
  \emph{IEEE Trans. Inform. Theory}, vol.~16, no.~4, pp. 452--459, 1970.

\bibitem{feng_generalized_1989}
G.-L. Feng and K.~K. Tzeng, ``{A Generalized Euclidean Algorithm for
  Multisequence Shift-Register Synthesis},'' \emph{IEEE Trans. Inform. Theory},
  vol.~35, no.~3, pp. 584--594, 1989.

\bibitem{hasse_theorie_1936}
H.~Hasse, ``{Theorie der höheren Differentiale in einem algebraischen
  Funktionenkörper mit vollkommenem Konstantenkörper bei beliebiger
  Charakteristik},'' \emph{Journal für die Reine und Angewandte Mathematik},
  no. 175, pp. 50--54, 1936.

\bibitem{krachkovsky_decoding_1997}
V.~Y. Krachkovsky and Y.~X. Lee, ``{Decoding for Iterative Reed--Solomon Coding
  Schemes},'' \emph{{IEEE} Trans. Magn.}, vol.~33, no.~5, pp. 2740--2742, 1997.

\bibitem{sidorenko_decoding_2008}
V.~R. Sidorenko, G.~Schmidt, and M.~Bossert, ``{Decoding Punctured
  Reed--Solomon Codes up to the Singleton Bound},'' in \emph{International
  {ITG} Conference on Source and Channel Coding ({SCC)}}, 2008.

\bibitem{guruswami_linear-algebraic_2011}
V.~Guruswami, ``{Linear-Algebraic List Decoding of Folded Reed--Solomon
  Codes},'' in \emph{{IEEE} Annual Conference on Computational Complexity
  ({CCC)}}, 2011, pp. 77--85.

\bibitem{feng_generalization_1991}
G.-L. Feng and K.~K. Tzeng, ``{A Generalization of the Berlekamp--Massey
  Algorithm for Multisequence Shift-Register Synthesis with Applications to
  Decoding Cyclic Codes},'' \emph{IEEE Trans. Inform. Theory}, vol.~37, no.~5,
  pp. 1274--1287, 1991.

\bibitem{zeh_fast_2011}
A.~Zeh and A.~Wachter, ``{Fast Multi-Sequence Shift-Register Synthesis with the
  Euclidean Algorithm},'' \emph{Adv. Math. of Comm.}, vol.~5, no.~4, pp.
  667--680, 2011.

\bibitem{macwilliams_theory_1988}
F.~J. MacWilliams and N.~J.~A. Sloane, \emph{{The Theory of Error-Correcting
  Codes}}.\hskip 1em plus 0.5em minus 0.4em\relax North Holland Publishing Co.,
  1988.

\bibitem{zeh_generalizing_2013}
A.~Zeh, A.~Wachter-Zeh, M.~Gadouleau, and S.~V. Bezzateev, ``{Generalizing
  Bounds on the Minimum Distance of Cyclic Codes Using Cyclic Product Codes},''
  in \emph{{IEEE} International Symposium on Information Theory ({ISIT)}},
  2013, pp. 126--130.

\bibitem{krachkovsky_decoding_1998}
V.~Y. Krachkovsky, ``{Decoding of Parallel Reed--Solomon Codes with
  Applications to Product and Concatenated Codes},'' in \emph{{IEEE}
  International Symposium on Information Theory ({ISIT})}, 1998, p.~55.

\bibitem{schmidt_collaborative_2009}
G.~Schmidt, V.~R. Sidorenko, and M.~Bossert, ``{Collaborative Decoding of
  Interleaved Reed--Solomon Codes and Concatenated Code Designs},'' \emph{IEEE
  Trans. Inform. Theory}, vol.~55, no.~7, pp. 2991--3012, 2009.

\bibitem{gopalan_list_2011}
P.~Gopalan, V.~Guruswami, and P.~Raghavendra, ``{List Decoding Tensor Products
  and Interleaved Codes},'' \emph{{SIAM} J. Comput.}, vol.~40, no.~5, pp.
  1432--1462, 2011.

\end{thebibliography}
\end{document}